\documentclass{theoretics}

\usepackage{mathtools}
\usepackage[shortcuts]{extdash}
\usepackage[capitalise,noabbrev,nameinlink]{cleveref}
\crefname{observation}{Observation}{Observations}
\usepackage[T1]{fontenc}

\title{Minimizing Tardy Processing Time on a Single Machine in Near-Linear Time}

\ThCSauthor[aff1]{Nick Fischer}{nick.fischer@insait.ai}[0009-0001-0909-3296]
\ThCSauthor[aff2]{Leo Wennmann}{wennmann@imada.sdu.dk}[0009-0001-3346-6494]
\ThCSaffil[aff1]{ INSAIT, Sofia University ``St.\ Kliment Ohridski'', Bulgaria }
\ThCSaffil[aff2]{ University of Southern Denmark, Odense, Denmark }
\ThCSthanks{This article was invited from ICALP 2024 \cite{FischerW24}.}
\ThCSshortnames{N.~Fischer and L.~Wennmann}  
\ThCSshorttitle{Minimizing Tardy Processing Time on a Single Machine in Near-Linear Time}
\ThCSyear{2025}
\ThCSarticlenum{14}
\ThCSreceived{Aug 15, 2024} 
\ThCSaccepted{Mar 30, 2025}
\ThCSpublished{Jul 17, 2025}
\ThCSdoicreatedtrue
\ThCSkeywords{Scheduling, Fine-Grained Complexity, Dynamic Strings}

\makeatletter
\newcommand\IfRestateTF{%
  \ifx\label\thmt@gobble@label 
    \expandafter\@firstoftwo
  \else
    \expandafter\@secondoftwo
  \fi
}
\makeatother
\newcommand{\RestateRemark}{\IfRestateTF{{\normalfont\bfseries (Restated) }}{}}

\definecolor{lightblue}{HTML}{8D9DB6}

\newcommand\Nat{\mathbb N}
\newcommand\Int{\mathbb Z}

\newcommand{\W}{X}

\DeclarePairedDelimiter\parens{\lparen}{\rparen}

\DeclarePairedDelimiter\set{\lbrace}{\rbrace}

\DeclarePairedDelimiterX\innerprod[2]{\langle}{\rangle}{#1,#2}

\DeclarePairedDelimiterX\intervaloo[2]{(}{)}{#1\,.\,.\,#2}
\DeclarePairedDelimiterX\intervaloc[2]{(}{]}{#1\,.\,.\,#2}
\DeclarePairedDelimiterX\intervalco[2]{[}{)}{#1\,.\,.\,#2}
\DeclarePairedDelimiterX\intervalcc[2]{[}{]}{#1\,.\,.\,#2}

\undef\Pr
\DeclareMathOperator*\Pr{\mathbf P}

\DeclareMathOperator\OPT{OPT}

\newcommand\Order{O}
\newcommand\order{o}
\newcommand\Oo{O}

\undef\mod
\undef\bmod
\DeclareMathOperator\bmod{mod}

\newcommand\mod[1]{\;\,({\bmod}\:#1)}

\newcommand{\Op}[1]{{\normalfont\texttt{#1}}}

\newcommand\MTPT{\smash{$1||\sum p_j U_j$}}
\newcommand\mMTPT{\smash{$P_m||\sum p_j U_j$}}

\begin{document}

\maketitle


\begin{abstract}
    In this work we revisit the elementary scheduling problem $1||\sum p_j U_j$. The goal is to select, among~$n$ jobs with processing times and due dates, a subset of jobs with maximum total processing time that can be scheduled in sequence without violating their due dates. This problem is NP-hard, but a classical algorithm by Lawler and Moore from the 60s solves this problem in pseudo-polynomial time~$\Order(nP)$, where~$P$ is the total processing time of all jobs. With the aim to develop best-possible pseudo-polynomial-time algorithms, a recent wave of results has improved Lawler and Moore's algorithm for~\smash{$1||\sum p_j U_j$}: First to time~\smash{$\widetilde\Order(P^{7/4})$}~[Bringmann, Fischer, Hermelin, Shabtay, Wellnitz; ICALP'20], then to time \smash{$\widetilde\Order(P^{5/3})$}~[Klein, Polak, Rohwedder; SODA'23], and finally to time~\smash{$\widetilde\Order(P^{7/5})$}~[Schieber, Sitaraman; WADS'23]. It remained an exciting open question whether these works can be improved further.
    
    In this work we develop an algorithm in near-linear time $\widetilde\Order(P)$ for the $1||\sum p_j U_j$ problem. This running time not only significantly improves upon the previous results, but also matches conditional lower bounds based on the Strong Exponential Time Hypothesis or the Set Cover Hypothesis and is therefore likely \emph{optimal} (up to subpolynomial factors). Our new algorithm also extends to the case of $m$ machines in time $\widetilde\Order(P^m)$. In contrast to the previous improvements, we take a different, more direct approach inspired by the recent reductions from Modular Subset Sum to dynamic string problems. We thereby arrive at a satisfyingly \emph{simple} algorithm.
\end{abstract}


\section{Introduction}
\label{sec:introduction}
Consider the following natural optimization problem: A worker is offered $n$ \emph{jobs}, where each job $j$ requires a \emph{processing time} of $p_j$ days and must be completed before some \emph{due date~$d_j$}. Which jobs should the worker take on in order to maximize their pay, assuming that the worker is paid a fixed amount per day of work? In standard scheduling notation~\cite{GrahamLLK79}, this task is somewhat cryptically called the ``$1||\sum p_j U_j$'' problem (see \cref{sec:preliminaries} for a formal definition). The $1||\sum p_j U_j$ problem constitutes an important scheduling task that is arguably among the \emph{simplest} nontrivial scheduling objectives, and has received considerable attention in the literature, especially in recent years.

The \MTPT{} problem naturally generalizes the famous Subset Sum problem,\footnote{Indeed, Subset Sum is the special case of \MTPT{} where all jobs share the same deadline $d$. In other words, Subset Sum is the $1|d_j = d|\sum p_j U_j$ problem.} and is therefore NP-hard. However, it does admit pseudo-polynomial-time algorithms---in 1969, Lawler and Moore~\cite{LawlerM69} pioneered the first such algorithm in time $\Order(n P)$, where~\makebox{$P = \sum_j p_j$} is the total processing time of all $n$ jobs. This result is the baseline in a line of research that, more than 50 years after the initial effort, is finally brought to a close in this paper.

\subparagraph{State of the Art.}
Lawler and Moore originally designed their algorithm for a weighted generalization of the \MTPT{} problem, and for this generalization the running time~\makebox{$\Order(nP)$} was proven to be conditionally tight.\footnote{In the so-called $1||\sum w_j U_j$ problem each job $j$ is rewarded by a specified pay~$w_j$ (instead of $p_j$). For this generalization the running time $\Order(nP)$ was proven to be conditionally optimal~\cite{CyganMWW19,KunnemannPS17}, in the sense that an algorithm in time~\smash{$\Order((n + P)^{2-\epsilon})$}, for any $\epsilon > 0$, contradicts the well-established $(\min,+)$-Convolution hypothesis. See also the discussion in~\cite{BringmannFHSW22}.} Even for the \MTPT{} problem the Lawler-Moore algorithm remained unchallenged for a long time. Only a few years ago, Bringmann, Fischer, Hermelin, Shabtay and Wellnitz~\cite{BringmannFHSW22} managed to solve \MTPT{} in time\footnote{We write~\smash{$\widetilde\Order(T) = T (\log T)^{\Order(1)}$} to suppress polylogarithmic factors.}~$\widetilde\Order(P^{7/4})$, showcasing that improvements over Lawler-Moore are indeed possible in certain parameter regimes (specifically, when~\makebox{$P \ll n^{4/3}$}). Their strategy is to design a reduction to an intermediate problem called \emph{Skewed Convolution}\footnote{Given length-$N$ integer vectors $A, B$, the Skewed Convolution problem is to compute the length-$(2N-1)$ vector $C$ defined by $C[k] = \min_{i+j=k} \max\set{A[i], B[j] - i}$.}, and to develop an~\smash{$\widetilde\Order(N^{7/4})$}-time algorithm for this intermediate problem.

Their work was later improved in two orthogonal ways. On the one hand, Klein, Polak and Rohwedder~\cite{KleinPR23} improved the running time of Skewed Convolution to~\smash{$\widetilde\Order(N^{5/3})$}. On the other hand, Schieber and Sitaraman~\cite{SchieberS23} improved the algorithmic reduction and established that, if Skewed Convolution is in time~\smash{$\widetilde\Order(N^\alpha)$}, then \MTPT{} is in time~\smash{$\widetilde\Order(P^{2-1/\alpha})$}. The state-of-the-art algorithm for \MTPT{} is obtained by combining these two works, resulting in time~\smash{$\widetilde\Order(P^{7/5})$}.

In contrast, fine-grained lower bounds for the Subset Sum problem rule out \MTPT{} algorithms in time~\smash{$\Order(P^{1-\epsilon} \cdot n^{\Order(1)})$}, for any $\epsilon > 0$, conditioned on either the influential Strong Exponential Time Hypothesis~\cite{AbboudBHS19} or the Set Cover Hypothesis~\cite{CyganDLMNOPSW16}. This leaves a substantial gap between the best known upper bound $\widetilde\Order(P^{7/5})$ and the conceivable optimum $\widetilde\Order(P)$. Closing this gap is the starting point of our paper:
\begin{center}
    \medskip
    \emph{Can the \MTPT{} problem be solved in near-linear time $\widetilde\Order(P)$?}
    \medskip
\end{center}
In light of the recent algorithmic developments~\cite{BringmannFHSW22,KleinPR23,SchieberS23}, a reasonable strategy appears to aim for even faster algorithms for the Skewed Convolution problem---unfortunately, this approach soon faces a barrier. Namely, improving the running time of Skewed Convolution beyond $\Order(N^{3/2})$ would entail a similarly fast algorithm for $(\max, \min)$-Convolution, which, while not ruled out under one of the big hypotheses, would be a surprising break-through in fine-grained complexity theory. This leaves us in an uncertain situation. Even if Skewed Convolution could be improved to time $\widetilde\Order(N^{3/2})$, this would mean that the \MTPT{} problem is in time $\widetilde\Order(P^{4/3})$~\cite{SchieberS23}. Are further improvements impossible?

\subparagraph{Our Results.}
In this paper we bypass this barrier and develop a new algorithm for \MTPT{} that avoids the reduction to Skewed Convolution altogether. We thereby successfully resolve our driving question:

\begin{restatable}{theorem}{MTPTalgorithm}
\label{theo:mtpt}\RestateRemark
The \MTPT{} problem can be solved in randomized time $\Order(P \log P)$ and in deterministic time~\smash{$\Order(P \log^{1+\order(1)} P)$}.
\end{restatable}

We stress that by the aforementioned lower bounds~\cite{AbboudBHS19,CyganDLMNOPSW16} our new algorithm is \emph{optimal}, up to lower-order factors, conditioned on the Strong Exponential Time Hypothesis or the Set Cover Hypothesis.

As an additional feature, and similar to all previous algorithms, our algorithm not only computes the optimal value of the given instance, but in fact reports for each value $0 \leq x \leq P$ whether there is a feasible schedule with processing time (i.e., pay) $x$. Moreover, we can compute an optimal schedule (represented as an ordered subset of jobs) in the same running time.

Another benefit of our work is that we managed to distill an astonishingly \emph{simple} algorithm. In fact, our algorithm is basically identical to the original Lawler-Moore algorithm, except that we replace certain naive computations by an appropriate efficient data structure on \emph{strings}, and employ a careful new analysis. This approach is inspired by the recent progress on the Modular Subset Sum problem~\cite{AxiotisBJTW19,AxiotisBBJNTW21,CardinalI21,Potepa21} (see \cref{sec:algorithm-for-one-machine} for more details). We~find it surprising that these conceptually simple ideas lead to near-optimal running times for~\makebox{$1||\sum p_j U_j$}. 

In particular, in contrast to previous improvements for \MTPT{}~\cite{BringmannFHSW22,KleinPR23,SchieberS23}, our algorithm is purely combinatorial and does not require the use of the Fast Fourier Transform. Given this simple nature of our algorithm, we are confident that actual implementations of the algorithm would perform well.

\subparagraph{Multiple Machines.}
The ``\mMTPT{}'' problem is the straightforward generalization of the \MTPT{} problem to $m$ workers (or machines) that can partition the jobs arbitrarily among themselves. The goal, as before, is to maximize the total workload across all workers while respecting all due dates. We assume for simplicity that $m$ is a constant.\footnote{When viewing $m$ as an input, it is easy to trace that our algorithms depend only polynomially on $m$.}

The Lawler-Moore algorithm generalizes in a straightforward manner to time $\Order(n P^m)$. For the algorithms based on Skewed Convolution, it seems significantly harder to derive multiple-machine generalizations.
Luckily, with some appropriate changes our new algorithm also generalizes to multiple machines:

\begin{restatable}{theorem}{ManyMachinesAlgorithm} \label{thm:many-machines}\RestateRemark
The \mMTPT{} problem can be solved in randomized time $\Order(P^m \log P)$ and in deterministic time~\smash{$\Order(P^m \log^{1+\order(1)} P)$}.
\end{restatable}

In particular, \cref{thm:many-machines} outperforms the Lawler-Moore algorithm by a near-linear factor~\smash{$\widetilde\Omega(n)$}. In contrast to the single-machine setting, however, we emphasize that this algorithm is not necessarily optimal. A conditional lower bound for this problem would, most likely, be derived from an analogous lower bound for the multiple-target Subset Sum problem~\cite{AntonopoulosPPV23}. This appears to be a challenging question which is not resolved yet.

\subparagraph{Alternative Parameters.}
So far we have only mentioned the parameters~$n$ and~\smash{$P = \sum_j p_j$}, which have been the main focus in previous work. But there are many other parameters worth considering. Natural candidates include the number of distinct deadlines ($D_\#$), the sum of all distinct deadlines ($D$), the largest processing time~($p_{\max} = \max_j p_j$) and the largest deadline~(\smash{$d_{\max} = \max_j d_j$}).

There has been research on developing nontrivial \MTPT{} algorithms (for a single machine) with respect to these parameters, such as an $\widetilde\Order(\min\set{P \cdot D_\#, P + D})$-time algorithm due to~\cite{BringmannFHSW22}, and an~\smash{$\widetilde\Order(n + p_{\max}^3)$}-time algorithm due to~\cite{KleinPR23}. We remark that the former is subsumed by our new results. The latter algorithm is incomparable to our result (specifically, our algorithm in time~\smash{$\widetilde\Order(P) = \widetilde\Order(n p_{\max})$} performs better if and only if $p_{\max} \gg n^{1/2}$). Both of these results~\cite{BringmannFHSW22,KleinPR23} generalize to $m$ machines as well, leading to similar comparisons with our~work.

It remains an interesting open question whether our $\widetilde\Order(P)$-time algorithm can be further improved with respect to the parameters $d_{\max}$ and $p_{\max}$. The Lawler-Moore algorithm achieves a running time of $O(n d_{\max})$, but in principle it seems reasonable that time $\widetilde\Order(n + d_{\max})$ can be achieved, as the analogous question for Subset Sum is resolved~\cite{Bringmann17,JinW19}. We leave this as an open question. An even more exciting question is whether we could possibly achieve time $\widetilde\Order(n + p_{\max})$. However, such an algorithm would entail a break-through for Subset Sum with small items, which currently seems out of reach.

\subparagraph{Further Related Work.}
This work is part of a bigger effort of the fine-grained complexity community to design best-possible pseudo-polynomial time algorithms for a host of optimization problems. This line of research encompasses, besides the aforementioned scheduling problems~\cite{BringmannFHSW22,HermelinMS22,KleinPR23,SchieberS23}, various variants of Subset Sum~\cite{KoiliarisX19,Bringmann17,AbboudBHS19,AxiotisBJTW19,AxiotisT19,BringmannW21,PolakRW21,AxiotisBBJNTW21,CardinalI21,Potepa21,DengMZ23}, Knapsack~\cite{Tamir09,CyganMWW19,KunnemannPS17,BateniHSS18,EisenbrandW20,PolakRW21,BringmannC23,ChenLMZ24,BringmannC22,Bringmann23,Jin23}, Integer Programming~\cite{EisenbrandW20,JansenR23} and many others~\cite{ChanH22,DengMZ23}.


\section{Preliminaries}
\label{sec:preliminaries}
Throughout, we write $[n] = \set{0, \dots, n-1}$ and use the interval notation $\intervalcc ij = \set{i, \dots, j}$, and similarly $\intervalco ij$, $\intervaloc ij$, $\intervaloo ij$. For two sets of integers $S, T$ and an integer $t$ we employ the sumset notation $S + t = \set{s + t : s \in S}$ and $S + T = \set{s + t : s \in S, t \in T}$.

\subparagraph{Scheduling Problems.}
We begin with a formal definition of the \MTPT{} problem. The input consists of $n$ \emph{jobs}, where each job $j \in [n]$ has a \emph{processing time $p_j \in \Nat_{> 0}$} and \emph{due dates~\makebox{$d_j \in \Nat_{> 0}$}}. A schedule is a permutation~\makebox{$\sigma : [n] \to [n]$}. The \emph{completion time $C_j$} of a job $j$ in the schedule~$\sigma$ is~\makebox{$C_j = \sum_{i : \sigma(i) \leq \sigma(j)} p_i$} (i.e., the total processing time of all jobs preceding $j$, including $j$ itself). We say that $j$ is \emph{early} if $C_j \leq d_j$ and \emph{tardy} otherwise, and let $U_j \in \set{0, 1}$ be the indicator variable indicating whether $j$ is tardy. In this notation, our objective is to find a schedule minimizing $\sum_j p_j U_j$ (i.e., the total processing time of all tardy jobs). This explains the description $1||\sum p_j U_j$ in three-field notation.\footnote{The $1$ in the first field denotes a single machine, the empty second field symbolizes no additional constraints, and the third field gives the objective to minimize~\smash{$\sum_j p_j U_j$}.} For convenience we have defined the problem in such a way that $p_j > 0$, and as a consequence we can always bound~\makebox{$n \leq P$}.\footnote{If jobs with processing time $p_j = 0$ were permitted, all of our algorithms would additionally require $\Order(n)$ time preprocessing.}

For the $m$-machine problem \mMTPT{} a schedule is analogously defined as a function~\makebox{$\sigma : [n] \to [n] \times [m]$}, where the first coordinate determines the order of the jobs as before, and the second coordinate determines the machine which is supposed to execute the job. The completion time $C_j$ is the total processing time of all jobs preceding $j$ \emph{on $j$'s machine} (including $j$ itself), and the objective of the problem remains unchanged. For simplicity, we assume throughout the paper that $m$ is a constant (it can easily be verified that we only omit~\smash{$m^{\Order(1)}$}-factors this way).

\subparagraph{Earliest-Due-Date-First Schedules.}
A key observation about \MTPT{} dating back to Lawler and Moore~\cite{LawlerM69} is that, without loss of generality, the early jobs are scheduled in increasing order of their due dates. This observation is leveraged as follows: We reorder the jobs such that $d_0 \leq \dots \leq d_{n-1}$ (we will stick to this ordering for the rest of the paper). Thus, the \MTPT{} problem is effectively to compute a subset of jobs $J \subseteq [n]$ that maximizes~\smash{$\sum_{j \in J} p_j$} and is \emph{feasible} in the sense that all jobs in $J$ are early (i.e., $C_j = \sum_{i \in J : i \leq j} p_i \leq d_j$ for all~\smash{$j \in J$}).

\subparagraph{Machine Model.}
We work in the standard Word RAM model with word size $\Theta(\log n + \log P)$ (such that each job can be stored in a constant number of cells). Moreover, all randomized algorithms mentioned throughout are Las Vegas (i.e., zero-error) algorithms running in their claimed time bounds with high probability $1 - 1/n^{\Omega(1)}$.


\section{Near-Optimal Algorithm for a Single Machine}
\label{sec:algorithm-for-one-machine}

In this section, we give the details of our near-optimal algorithm for \MTPT{}. We start with a brief summary of the Lawler-Moore algorithm. 

\subparagraph{Lawler and Moore's Baseline.}
The Lawler-Moore algorithm~\cite{LawlerM69} is the natural dynamic programming solution for the $1||\sum p_j U_j$ problem. We present it here by recursively defining the following sets~\makebox{$S_0, \dots, S_n \subseteq \intervalcc 0P$}:
\begin{alignat*}{2}
    S_0 &= \set{0}, \\
    S'_{j+1} &=  S_{j} + \set{0, p_j} \qquad && (j \in [n]), \\
    S_{j+1} &= S'_{j+1} \cap [0,d_{j}]  \qquad && (j \in [n]).
\end{alignat*}
(The construction of $S_{j+1}$ is divided into two steps as this will be convenient later on.) Each set $S_{j+1}$ can naively be computed from $S_j$ in time $\Order(P)$, and thus all sets $S_0, \dots, S_n$ can be naively computed in time $\Order(nP)$. We can ultimately read off the optimal value as $\max S_n$, based on the following observation:

\begin{observation}[Lawler and Moore \cite{LawlerM69}]
    \label{obs:s-sets}
    There is a feasible schedule of total processing time~$t$ if and only if $t \in S_n$.
\end{observation}

More generally, $S_{j+1}$ is the set of processing times of feasible schedules involving the jobs~$0, \dots, j$. To see this, consider any feasible schedule of the jobs $0, \dots, j-1$ (whose processing time is in $S_j$). We can either leave out the next job $j$ or append to the schedule. Thereby, the set of processing times becomes $S'_{j+1} = S_j + \set{0, p_j} = \set{s, s + p_j : s \in S'_{j}}$. However, this appended schedule is not necessarily feasible as it might not comply with the due date~$d_j$. Hence, all processing times greater than $d_j$ are deleted again in the construction of $S_{j+1}$.

\subparagraph{Our Approach.}
Perhaps surprisingly, our algorithm essentially follows the same approach, i.e., our goal remains to compute the sets $S_0, \dots, S_n$. However, we will demonstrate that the naive $\Order(P)$-time computation of each step can be significantly sped up. Our algorithm relies on two ingredients---an algorithmic and a structural one.

\subparagraph{Ingredient 1: An Efficient Data Structure.}
As the sets $S_j'$ and $S_j$ are constructed in a highly structured way, can we compute them faster than time $\Order(P)$? Specifically, is there a way to~(i)~compute each set $S'_{j+1}$ in time proportional to the number of \emph{insertions} $|S'_{j+1} \setminus S_j|$, and to~(ii)~compute $S_{j+1}$ in time proportional to the number of \emph{deletions} $|S'_{j+1} \setminus S_{j+1}|$?

Question (i) is closely related to the Subset Sum problem, and has been successfully resolved in~\cite{AxiotisBJTW19} leading to near-optimal algorithms for \emph{Modular} Subset Sum. Their solution based on linear sketching is quite involved~\cite{AxiotisBJTW19}, but two independent papers~\cite{AxiotisBBJNTW21,CardinalI21} provided a significantly simpler proof by replacing linear sketching with a reduction to a dynamic string problem; see also~\cite{Potepa21}\footnote{In~\cite{Potepa21}, Pot\k{e}pa proposes an improved deterministic data structure with applications to the Modular Subset Sum problem. A priori, it looks like their improvement might similarly apply to our setting. Unfortunately, the data structure is only efficient if we have the freedom to arbitrarily reorder the items, which is prohibitive for us as we have to stick to the order $d_0 \leq \dots \leq d_{n-1}$.}. Regarding (ii), it turns out that we can adapt this reduction to the dynamic string problem to efficiently accommodate our deletions. The following lemma summarizes the resulting data structure; we defer its proof to \cref{subsec:dynamic-string-data-structure}.

\begin{lemma}[Sum-Cap Data Structure]
    \label{lem:set_ds}
    There is a randomized data structure that maintains a set $S \subseteq [u]$ and supports the following operations:
    \smallskip
    \begin{itemize} 
        \item \parbox{2.2cm}{$\Op{init}(S)$:} Initializes the data structure to the given set $S \subseteq [u]$.\\\parbox{2.2cm}{~} Runs in time $\Oo(|S| \cdot \log u + \log^2 u)$.
        
        \item \parbox{2.2cm}{$\Op{query}(s)$:} Given $s \in [u]$, tests whether $s \in S$.\\\parbox{2.2cm}{~} Runs in time $\Oo(\log u)$.
        
        \item \parbox{2.2cm}{$\Op{sum}(p)$:} Given $p \in [u]$, updates $S \gets S + \set{0, p}$.\\\parbox{2.2cm}{~} Runs in time $\Oo(|(S+p)\setminus S| \cdot \log u)$ (where $S$ is as before the operation).

        \item \parbox{2.2cm}{$\Op{cap}(d)$:} Given $d \in [u]$, updates $S \gets S \cap [d]$.\\\parbox{2.2cm}{~} Runs in time $\Oo(\log u)$.
    \end{itemize}
    \smallskip
    If at any point during the execution an element $s \notin [u]$ is attempted to be inserted, the data structure becomes invalid. Moreover, the data structure can be made deterministic at the cost of worsening all operations by a factor $\log^{\order(1)} u$.
\end{lemma}

\subparagraph{Ingredient 2: A Structural Insight.}
What have we gained by computing the sets $S_j$ and~$S_j'$ with the data structure from \cref{lem:set_ds}? 
Due to the particularly efficient \Op{cap} operation, the computation of the sets $S_j$ is essentially for free.
Computing the sets $S_j'$ using the \Op{sum} operation, however, amounts to time
\begin{equation*}
    \widetilde\Order\parens*{\sum_{j \in [n]} |S_{j+1}' \setminus S_j|}.
\end{equation*}
A priori, it is not clear whether this is helpful. In case of only \emph{inserting} elements, this sum could be conveniently bounded by $P$ (as is the case for Modular Subset Sum). Unfortunately, we additionally have to deal with \emph{deletions}. Specifically, it is possible that some element $s$ is inserted in~$S_1'$, deleted in~$S_1$, inserted again in~$S_2'$, and so on. All in all, $s$ could be inserted up to $n$ times, and so the only immediate upper bound for the sum is~$nP$ (which recovers the Lawler-Moore running time).

Our crucial structural insight is that, while the same element can indeed be inserted and deleted multiple times, the \emph{total number of insertions} is nevertheless bounded. More precisely, we show that the overall number of insertions is at most $\Order(P)$:

\begin{lemma}[Bounded Insertions]
    \label{lem:s-sets}
    It holds that $\sum_{j\in[n]} |S'_{j+1} \setminus S_{j}| \le 2P + 1$.
\end{lemma}
\begin{proof}
    We split the sum into two parts:
    \[
        \sum_{j\in[n]} \big|(S'_{j+1} \setminus S_{j})\big| = \sum_{j\in[n]} \big|(S'_{j+1} \setminus S_{j}) \cap \intervalcc{0}{d_j}\big| + \sum_{j\in[n]} \big|(S'_{j+1} \setminus S_{j}) \cap \intervaloc{d_j}{P}\big|.
    \]
    Intuitively, the first sum counts the number of elements that are irreversibly inserted into the sets~\makebox{$S_{j+1}, \dots, S_n$} in the $j$-th step.
    The second sum counts the number of elements that are inserted into $S_{j+1}'$ and immediately deleted in $S_{j+1}$.

    For the first sum, consider the following observation: For any $x \in \intervalcc{0}{P}$, if $x \le d_j$ and $x \in S'_{j+1} \setminus S_{j}$, then $x \in S_{j+1}, \dots , S_n$ (since $d_j \le d_{j+1}, \dots, d_{n-1}$).
    It follows that
    \[
        \big|\set{ j \in [n] : x \in (S'_{j+1} \setminus S_{j}) \cap \intervalcc{0}{d_j}}\big| \le 1
    \]
    for all $x \in \intervalcc{0}{P}$.
    Thus,
    \begin{align*}
        \sum_{j\in[n]} \big|(S'_{j+1} \setminus S_{j}) \cap \intervalcc{0}{d_j}\big| 
        =  \sum_{x\in\intervalcc{0}{P}} \big|\set{ j \in [n] : x \in (S'_{j+1} \setminus S_{j}) \cap \intervalcc{0}{d_j}}\big| 
        \le P + 1.
    \end{align*}   

    Second, we bound $|(S'_{j+1} \setminus S_{j}) \cap \intervaloc{d_j}{P}| \le |S'_{j+1} \cap \intervaloc{d_j}{P}|$.
    Note that $S_{j} \subseteq \intervalcc{0}{d_{j-1}}$ and therefore $S_{j+1}' = S_{j} + \set{0, p_j} \subseteq \intervalcc{0}{d_{j-1} + p_j}$.
    Consequently,
    \begin{align*}
        \big|S'_j \cap \intervaloc{d_j}{P} \big|
        \le \big| \intervalcc{0}{d_{j-1} + p_j} \cap  \intervaloc{d_j}{P} \big| 
        \le d_{j-1} + p_j - d_j 
        \le p_j,
    \end{align*}
    where the final inequality follows from the assumption that $d_{j-1} \le d_j$.
    Hence, the number of overall deletions is
    \[
        \sum_{j\in[n]} \big|(S'_{j+1} \setminus S_{j}) \cap \intervaloc{d_j}{P}\big| \le \sum_{j\in[n]} p_j = P.
    \]
    Combining both parts concludes the proof.
\end{proof}

The proof for \cref{lem:set_ds} is provided in \cref{subsec:dynamic-string-data-structure}.
Using \cref{lem:s-sets} and \ref{lem:set_ds}, we are in the position to show our main result.

\MTPTalgorithm*
\begin{proof}
    In summary, our algorithm works as follows. We compute~\makebox{$S_0, \dots, S_n \subseteq \intervalcc 0P$} using the data structure from \cref{lem:set_ds} (with $u = P + 1$).
    Specifically, after initializing $S$ with \Op{init}$(S_0)$, we repeatedly construct the sets $S_j'$ and $S_j$ using the operations \Op{sum}$(p_j)$ and \Op{cap}$(d_j)$ for all~\makebox{$j \gets 0, \dots, n-1$}.
    The largest element in the final set $S = S_n$ is the maximal total processing of a feasible schedule of all jobs $0, \dots, n-1$.
    Finding and returning it is the last step of our algorithm, by repeatedly using the \Op{query}$(i)$ operation and returning the largest index~$i$ for which the query returns yes.
    The correctness of our algorithm follows from \cref{obs:s-sets}.

    The running time is composed of the following parts:
    The initialization runs in time $\Oo(\log^2 P)$, the repeated use of \Op{sum} and \Op{cap} takes time $\Oo(\sum_{j\in [n]} (|S_{j+1}' \setminus S_j| \cdot \log P + \log P))$ and the optimal value is found in time~$\Oo(P)$.
    Using~$n \le P$ and \cref{lem:s-sets}, it holds that 
    \[
        \Oo\parens*{\sum_{j\in [n]} (|S_{j+1}' \setminus S_j| \cdot \log P + \log P)} \le \Oo(P\log P).
    \]
    In total, we have a randomized running time of $\Oo(P\log P)$.
    Applying the same arguments yields the deterministic running time of $\Oo(P \log^{1 + o(1)}P)$.
\end{proof}

We stress that the algorithm described in this section only computes the optimal~\emph{value}. 
In \cref{subsec:schedule}, we explain how our algorithm can be easily extended to obtain the optimal \emph{schedule} as well.


\subsection{Cap-Sum Data Structure via Dynamic Strings}
\label{subsec:dynamic-string-data-structure}
In this section, we provide the missing proof of \cref{lem:set_ds} by a reduction to the \emph{dynamic strings} data structure problem. This is the fundamental problem of maintaining a collection of strings that can be concatenated, split, updated, and tested for equality---see~\cite{SundarT94,PughT89,MehlhornSU97,AlstrupBR00,GawrychowskiKKL18}. We summarize the state of the art in the following lemma; the fastest randomized (and in fact, optimal) data structure is due to Gawrychowski, Karczmarz, Kociumaka, Lacki and Sankowski~\cite{GawrychowskiKKL18}, and for the fastest deterministic one see~\cite[Section~8]{KempaK22}.

Here we use standard string notation for a string $x$, where $x[i]$ denotes the letter at index~$i$, and $x\intervalcc{i}{j}, x \intervalco ij$ denote the appropriate substrings.

\begin{lemma}[Dynamic String Data Structure \cite{GawrychowskiKKL18,KempaK22}]
    \label{lem:dynamic_string_ds}
    There is a data structure that maintains a dynamic collection $\W$ of non-empty strings and support the following operations:
    \smallskip
    \begin{itemize}
        \item \parbox{3.75cm}{$\Op{make\_string}(x)$:} Given any string $x \in \Sigma^+$, inserts $x$ into $X$.

        \item \parbox{3.75cm}{$\Op{concat}(x_1,x_2)$:} Given $x_1, x_2 \in X$, inserts the concatenation $x_1 x_2$ into $\W$.

        \item \parbox{3.75cm}{$\Op{split}(x,i)$:} Given $x \in X$ and $i \in \intervalco{0}{|x|}$, inserts $x\intervalcc{0}{i}$ and $x\intervaloo{i}{|x|}$ into $X$.

        \item \parbox{3.75cm}{$\Op{LCP}(x_1,x_2)$:} Given $x_1, x_2 \in X$, returns the length $\ell$ of their longest common \\\parbox{3.75cm}{~} prefix, i.e., returns $\max\set{0 \le \ell \leq \min\set{|x_1|, |x_2|} : x_1\intervalco{0}{\ell} = x_2\intervalco{0}{\ell}}$.

        \item \parbox{3.75cm}{$\Op{query}(x,i)$:} Given $x \in X$ and $i \in [|x|]$, returns $x[i]$.
    \end{itemize}
    \smallskip
    Let $n$ be the maximum of the total length of all strings and the number of executed operations. 
    Then all operations run in randomized time $\Oo(\log n)$ or in deterministic time $\Oo(\log^{1 + o(1)}n)$, except for \Op{make\_string} which takes time $\Oo(|x| + \log n)$ and $\Oo(|x| \cdot \log^{o(1)}n)$, respectively.
\end{lemma}

For the sake of convenience, we include two more dynamic string operations that are derived from the previous lemma in a black-box fashion.
As both are standard operations~\cite{GawrychowskiKKL18}, we only provide their implementations for completeness. 
\smallskip
\begin{itemize}
    \item $\Op{update}(x, i, \sigma)$: Given $x \in X$, an index $i \in [|x|]$ and $\sigma \in \Sigma$, inserts the string obtained from $x$ by changing the $i$-th character to $\sigma$ into the data structure. To implement this, we split the string $x$ twice to separate the letter $x[i]$ from the rest of the string.
    Specifically, we obtain the substring $x\intervalcc{0}{i}$ using \Op{split}$(x,i)$ and further divide it to get the substring~$x\intervalco{0}{i}$ by \Op{split}$(x\intervalco{0}{i}, i-1)$.
    Next, the \Op{make\_string}$(\sigma)$ operation creates the string~$\sigma$.
    Lastly, we use $\Op{concat}(x\intervalco{0}{i}, \Op{concat}(\sigma, x\intervaloo{i}{|x|}))$ to reinsert $\sigma$ between the two substrings.
    \item $\Op{LCE}(x_1, x_2, i_1,i_2)$: Given $x_1, x_2 \in X$ and $i_1 \in [|x_1|], i_2 \in [|x_2|]$, returns the longest common extension $\max\set{0 \le \ell \leq \min\set{|x_1| - i_1, |x_2| - i_2} : x_1\intervalco{i_1}{i_1 + \ell} = x_2\intervalco{i_2}{i_2 + \ell}}$. To implement this, using the two operations \Op{split}$(x_1,i-1)$ and \Op{split}$(x_2,j-1)$ we first separate the substrings $x_1\intervalco{i}{|x_1|}$ and $x_2\intervalco{j}{ |x_2|}$.
    Observe that the length of the longest common extension of the original strings is exactly the length of the longest common prefix of $x_1\intervalco{i}{|x_1|}$ and $x_2\intervalco{j}{ |x_2|}$ returned by the operation \Op{LCP}$(x_1\intervalco{i}{|x_1|}, x_2\intervalco{j}{ |x_2|})$.
\end{itemize}
\smallskip
Both \Op{update} and \Op{LCE} require a constant number of original operations that run in randomized time $\Oo(\log n)$, or deterministic time $\Oo(\log^{1 + o(1)}n)$.

Now, we are in the position to provide the proof of \cref{lem:set_ds}. Recall that this proof is in parts borrowed from~\cite{AxiotisBBJNTW21,CardinalI21}. 

\begin{proof}[Proof of {Lemma~\ref{lem:set_ds}}]
    We maintain the set $S\subseteq [u]$ as an indicator string $x_S \in \set{0,1}^{u}$ such that $i \in S$ if and only if $x_S[i] = 1$.
    \smallskip
    \begin{itemize}
        \setlength\parindent{1.5em}
        \item \Op{init}$(S)$: 
        Using repeated squaring, we construct the string $0^u$ by inserting $w=0$ and concatenating $w$ with itself $\log u$ times.
        Note that $0^u$ will remain in the data structure and can be used by other operations.
        We compute $x_S$ by updating $0^u$ using \Op{update}$(x_S, i, 1)$ for all indices $i \in S$.
        Since the repeated squaring takes time $\Oo(\log^2 u)$ and updating the elements takes time $\Oo(|S| \cdot \log u )$, the \Op{init} operation runs in time $\Oo(\log^2 u +|S| \cdot \log u )$.
        
        \item \Op{query}$(i)$: 
        As the original data structure already provides a query operation, we use \Op{query}$(x_S,i)$ that returns $x_S[i]$ in time $\Oo(\log u )$.

        \item \Op{sum}$(p)$: 
        We implement \Op{sum} in three steps.
        First, we will compute the string $x_{S+p}$ that represents the set $S +p$.
        Observe that $x_{S+p}$ is obtained by shifting $x_{S}$ by $p$ positions to the right.
        Thus, we extend $x_S$ using \Op{concat}$(0^{p}, x_S)$ where the string $0^p$ is split off the precomputed string $0^u$ using \Op{split}$(0^u,p-1)$.
        Then we trim it down to length~$u$ with \Op{split}$(0^{p} x_S, u-1)$. (Note that due to the assumption $\max(S + p) < u$ we only split off zeros in this step.)
        
        Second, note that the desired string $x_{S\cup (S+p)}$ is the result of the bit-wise OR of $x_S$ and $x_{S+p}$.
        We compute the set $D = \set{i \in [u] : x_S[i] \neq x_{S+p}[i]}$ that contains all indices at which $x_S$ and $x_{S+p}$ differ from each other.
        To this end, starting with $i \gets 0$, we will repeat the following process as long as $i < u$: Compute $\ell \leftarrow \Op{LCE}(x_S,x_{S+p},i,i)$ to determine the next index $i + \ell$ at which both strings differ, insert $D \leftarrow D \cup \set{i+\ell}$ and update $i \leftarrow i + \ell + 1$. 
        
        As the third and last step, we compute $x_{S\cup (S+p)}$ by updating $x_S$ using \Op{update}$(x_S, i, 1)$ for all indices $i \in D$.
        
        Creating the shifted string $x_{S+p}$ takes time $\Oo(\log u)$.
        Both the construction of set~$D$ and computing the string~$x_{S\cup (S+p)}$ require $|(S+p)\setminus S| + |S\setminus (S+p)|$ many operations that each run in time $\Oo(\log u)$.
        Since $|S+p| = |S|$, we have 
        \[
            |(S+p)\setminus S| = |S+p| - |(S+p) \cap S| = |S| -  |(S+p) \cap S| = |S\setminus (S+p)|
        \]
        and therefore $|(S+p)\setminus S| + |S\setminus (S+p)| = 2 \cdot |(S+p)\setminus S|$.
        In summary, the \Op{sum} operation takes time $\Oo(|(S+p)\setminus S| \cdot \log u)$.
        
        \item \Op{cap}$(d)$: 
        In order to set $x_S[i] = 0$ for all $i \in \intervaloo{d}{u}$, we separate the substring $x_S\intervalcc{0}{d}$ with \Op{split}$(x_S,d)$ and split the substring $0^{u-d-1}$ off the precomputed string $0^u$ using \Op{split}$(0^u,u-d-2)$.
        Then $x_S$ is assembled using \Op{concat}$(x_S\intervalcc{0}{d},0^{u-d-1})$.
        All three operations take time $\Oo(\log u)$.
    \end{itemize}
    \smallskip
    Following \cref{lem:dynamic_string_ds}, the deterministic running times can be obtained by worsening all operations by a factor $\log^{\order(1)} u$.
\end{proof}


\subsection{Obtaining an Optimal Schedule} \label{subsec:schedule}
In the previous sections we have argued that the optimal value $\OPT$ (i.e., the maximum total processing of a feasible schedule) can be computed in near-linear time $\widetilde\Order(P)$. In this section we explain how the actual optimal schedule can be computed by a deterministic post-processing routine in time $\Order(n)$.

The idea is, as is standard for dynamic programming algorithms, to trace through the sets~\makebox{$S_0, \dots, S_n$} in reverse order. To make this traversal efficient, we slightly modify our algorithm to additionally compute an array $A \intervalcc0P$ such that $A[s] = \min\set{j \in [n] : s \in S_{j+1}}$. Intuitively,~$A[s]$ stores the smallest job $j$ such that there exists a feasible schedule with total processing time $s$ that contains $j$ and a subset of the jobs $\set{0, \dots, j-1}$. It is easy to appropriately maintain the array $A$ whenever an element is inserted into $S_{j+1}'$ without worsening the asymptotic running time. Then, in order to compute an optimal schedule, we apply the following algorithm: We initialize $J \gets \emptyset$ and~\smash{$s \gets \OPT$}. We repeatedly retrieve the next job~\makebox{$j \gets A[s]$} and update~\makebox{$J \gets J \cup \set{j}$} and $s \gets s - p_j$, until $s = 0$. In each step, we identify a job $j$ that is contained in the optimal schedule, and thus $J$ is an optimal schedule once the process has terminated. In fact, the same idea can be used to retrieve a feasible schedule for \emph{any} given processing time $s \in S_n$.


\section{Generalization to Multiple Machines}
\label{sec:manymachines}

In this section, we show that our algorithm for \MTPT{} can be extended to~\mMTPT{}. Since it follows the same approach as the single machine algorithm, we will keep this section short and concise.
For more details refer to \cref{sec:algorithm-for-one-machine}.

Let $e_0, \dots, e_{m-1}$ denote the standard unit vectors of~$\Int^m$, then we recursively define the sets $S_0, \dots, S_n \subseteq [0, P]^m$ as follows:
\begin{alignat*}{2}
    S_0 &= \set{0}, \\
    S_{j+1}' &= S_{j} + \set{0, p_j \cdot e_0, \dots, p_j \cdot e_{m-1}} \qquad &&(j \in [n]),  \\
    S_{j+1} &= S_{j+1}' \cap [0, d_j]^m \qquad &&(j \in [n]). 
\end{alignat*}
As before, the optimal value is the maximum entry in $S_n$:

\begin{observation}[Lawler and Moore~\cite{LawlerM69}]
    \label{obs:generalized-s-sets}
    There is a feasible schedule of total processing time $t$ if and only if there is some $s \in S_n$ with $s_0 + \dots + s_{m-1} = t$.
\end{observation}

The crucial difference to before is that here all $s \in S_{j}$ are vectors where their $i$-th entry corresponds to the $i$-th machine.
Because each job is either scheduled on exactly one machine or not at all, we consider all scheduling possibilities of job $j$ with $S_{j} + \set{0, p_j \cdot e_0, \dots, p_j \cdot e_{m-1}}$.
As our goal is again to bound the total number of insertions, see the following lemma:

\begin{lemma}[Generalized Bounded Insertions]
    \label{lem:generalized-bounded-insertions}
    It holds that
    \[
        \sum_{j \in [n]} \big|S_{j+1}' \setminus S_{j}\big| \leq (m+1)  \cdot (P+1)^m.
    \]
\end{lemma}
\begin{proof}
    In the following, we consider two parts of the sum separately:
    \[
        \sum_{j\in[n]} \big|(S'_{j+1} \setminus S_{j})\big| = \sum_{j\in[n]} \big|(S'_{j+1} \setminus S_{j}) \cap \intervalcc{0}{d_{j}}^m\big| + \sum_{j\in[n]} \big|(S'_{j+1} \setminus S_{j}) \setminus \intervalcc{0}{d_{j}}^m\big|.
    \]
    In other words, in analogy to \cref{lem:s-sets}, we first count the number of elements that are irrevocably inserted into $S_{j+1}, \dots, S_n$ in the $j$-th step.
    Second, we count the number of elements that are inserted into $S_{j+1}'$ and instantly deleted in $S_{j+1}$.

    We bound the first sum with the following observation.
    For any $x \in \intervalcc{0}{P}^m$, it holds that if $x \in (S'_{j+1} \setminus S_{j})$ and $x \in \intervalcc{0}{d_j}$, then $x \in S_{j+1}, \dots, S_n$. 
    This follows directly from the assumption that $d_j \le d_{j+1}, \dots ,d_{n-1}$.
    Therefore, it holds that
    \[
        \big|\set{ j \in [n] : x \in (S'_{j+1} \setminus S_{j}) \cap \intervalcc{0}{d_j}^m}\big|  \le 1,
    \]
    for all $x \in \intervalcc{0}{P}^m$, and thus
    \begin{align*}
        \sum_{j\in[n]} \big|(S'_{j+1} \!\setminus\! S_{j}) \cap \intervalcc{0}{d_j}^m \big| 
        &= \!\!\!\!\!\sum_{x\in\intervalcc{0}{P}^m}\!\!\!\!\! \big|\set{ j \in [n] : x \in (S'_{j+1} \!\setminus\! S_{j}) \cap \intervalcc{0}{d_j}^m}\big| \le (P + 1)^m.
    \end{align*}   

    For the second sum, we bound $|(S'_{j+1} \setminus S_{j}) \setminus \intervalcc{0}{d_j}^m | \le |S_{j+1}' \setminus \intervalcc{0}{d_j}^m |$.
    Using the fact that~\makebox{$S_{j} \subseteq \intervalcc{0}{d_{j-1}}^m$}, we have that 
    \begin{align*}
        S_{j+1}' 
        &= S_{j} + \set{0, p_j \cdot e_0, \dots, p_j\cdot e_{m-1}} \\
        &\subseteq \intervalcc{0}{d_{j-1}}^m + \set{0, p_j \cdot e_0, \dots, p_j \cdot e_{m-1}} \eqqcolon V_{j+1}.
    \end{align*}
    As each job $j$ is scheduled on exactly one machine, we observe that $V_{j+1}$ is the set of vectors where all entries are in $\intervalcc{0}{d_{j-1}}$, except for possibly one entry that is in $\intervalcc{0}{d_{j-1} + p_j}$.
    Hence,~\makebox{$V_{j+1} \setminus  \intervalcc{0}{d_j}^m$} is the set of vectors where all entries are in $\intervalcc{0}{d_{j-1}}$, except for exactly one entry that is in $\intervaloc{d_j}{d_{j-1} + p_j}$.
    Next, we bound the size of $V_{j+1} \setminus \intervalcc{0}{d_j}^m$: There are $m$ options for the index of the special entry, there are $d_{j-1} + p_j - d_j$ options for the value of the special entry, and finally there are $(P+1)^{m-1}$ options for the other $m-1$ entries.
    Thus, 
    \begin{align*}
        \big|S_{j+1}' \setminus \intervalcc{0}{d_j}^m \big|
        &\le \big| V_{j+1} \setminus  \intervalcc{0}{d_j}^m \big| \\
        &\le m \cdot (P+1)^{m-1} \cdot (d_{j-1} + p_j - d_j) \\
        &\le m \cdot (P+1)^{m-1} \cdot p_j,
    \end{align*}
    where the final inequality follows from the assumption that $d_{j-1} \le d_j$.
    Consequently, the second sum is bounded by
    \[
        \sum_{j\in[n]} \big|S_{j+1}' \setminus \intervalcc{0}{d_j}^m \big| \le m\cdot(P+1)^{m-1} \cdot \sum_{j\in[n]} p_j \leq m \cdot (P+1)^m.
    \]
    Combining the bounds for both sums yields the overall bound.
\end{proof} 

Analogous to \cref{sec:algorithm-for-one-machine}, we use the generalized sum-cap data structure to efficiently maintain the generalized sets $S_0, \dots, S_n$.

\begin{lemma}[Generalized Sum-Cap Data Structure]
    \label{lem:generalized-set-ds}
    There is a randomized data structure maintaining a set $S \subseteq [u]^m$ that supports the following operations:
    \smallskip
    \begin{itemize}
        \item \parbox{2.2cm}{$\Op{init}(S)$:} Initializes the data structure to the given set $S \subseteq [u]^m$.\\\parbox{2.2cm}{~} Runs in time $\Oo(\log^2 u +|S| \cdot \log u )$.
        
        \item \parbox{2.2cm}{$\Op{query}(s)$:} Given $s \in [u]^m$, tests whether $s \in S$.\\\parbox{2.2cm}{~} Runs in time $\Order(\log u)$.
        
        \item \parbox{2.2cm}{$\Op{sum}(T)$:} Given $T \subseteq [u]^m$, updates $S \gets S + T$.\\\parbox{2.2cm}{~} Runs in time $\Order(|T| \cdot |(S + T) \setminus S| \cdot \log u)$ (where $S$ is as before the operation).

        \item \parbox{2.2cm}{$\Op{cap}(d)$:} Given $d \in [u]$, updates $S \gets S \cap [d]^m$.\\\parbox{2.2cm}{~} Runs in time $\Order(u^{m-1} \cdot \log u)$.
    \end{itemize}
    \smallskip
    If at any point during the execution an element $s \not\in [u]^m$ is attempted to be inserted, the data structure becomes invalid.
    Moreover, the data structure can be made deterministic at the cost of worsening all operations by a factor $\log^{\order(1)} u$.
\end{lemma}
\begin{proof}
    Let $U = u^m$. Let $\phi:[u]^m \rightarrow [U]$ be the bijection defined by $\phi(s) = \sum_{i \in [m]} s_i u^i$. We extend the definition to sets~\makebox{$S \subseteq [u]^m$} via $\phi(S) = \set{\phi(s) : s \in S}$. We maintain the set~\makebox{$S\subseteq [u]^m$} as the indicator string of $\phi(S)$, namely $x_{\phi(S)} \in \set{0,1}^U$, such that $i \in S$ if and only if $x_{\phi(S)}[\phi(i)] = 1$. In other words, we store for each $s \in [u]^m$, listed in lexicographical order, whether $s \in S$.
    \smallskip
    \begin{itemize}
        \setlength\parindent{1.5em}
        \item \Op{init}$(S)$: 
        The string $0^{U}$ is constructed using repeated squaring by inserting $x=0$ and concatenating $x$ with itself $m \log u$ times.
        It will remain in the data structure available to other operations.
        Repeatedly using \Op{update}$(\cdot, \phi(s), 1)$ for all entries $s \in S$, we update the string $0^{U}$ to obtain $x_{\phi(S)}$.
        As repeated squaring takes time $\Oo(\log^2 u)$ and updating the elements takes time $\Oo(|S| \cdot \log u )$, the \Op{init} operation runs in time $\Oo(\log^2 u +|S| \cdot \log u)$.
        
        \item \Op{query}$(s)$: 
        Using \Op{query}$(x_{\phi(S)},s)$ from the string data structure allows us to return~$x_{\phi(S)}[s]$ in time $\Oo(\log u)$.
        
        \item \Op{sum}$(T)$: We can assume without loss of generality that $0 \in T$ (as otherwise we can simply shift $T$ and $x_{\phi(S)}$ appropriately). Fix an arbitrary nonzero element $p \in T$. Analogously to \cref{lem:set_ds}, we first show how to compute the set $(S + p) \setminus S$ in output-sensitive time.
        The string~$x_{\phi(S+p)}$ representing the set $S +p$ can be computed using the two following facts.
        If \makebox{$x, y, x+y \in [u]^m$}, then it holds that $\phi(x+y) = \phi(x) + \phi(y)$.
        Further, if $p \in [u]^m$ and $S, S+p \subseteq [u]^m$, then~\makebox{$\phi(S+p) = \phi(S) + \phi(p)$}.
        Therefore, $x_{\phi(S+p)}$ is $x_{\phi(S)}$ up to a shift of $\phi(p)$, and can be obtained by \Op{split}$(\Op{concat}(0^{\phi(p)}, x_{\phi(S)}), U-1)$.
        
        Similarly to \cref{lem:set_ds}, repeatedly using \Op{LCE} queries allows us to first compute the set $\set{i \subseteq [U] : x_{\phi(S)}[i] \neq x_{\phi(S+p)}[i]}$ and then read off the symmetric difference of $S$ and~\makebox{$S + p$}, denoted by $D_p = (S \setminus (S + p)) \cup ((S + p) \setminus S)$.
        We repeat the same for all other nonzero elements $p \in T$.
        Let $D = \bigcup_{p \in T \setminus \set{0}} D_p$, then we have $D \supseteq (S + T) \setminus S$. Thus, we can update the indicator string $x_{\phi(S)}$ by calling $\Op{update}(\cdot, \phi(s), 1)$ for all $s \in D$.

        In order to bound the running time, we bound the size of the sets $D_p$ and $D$. Using that~\makebox{$|S \setminus (S + p)| = |(S + p) \setminus S|$}, we have that $|D_p| \leq 2 |(S + p) \setminus S| \leq 2 |(S + T) \setminus S|$, and therefore $|D| \leq 2|T| \cdot |(S + T) \setminus S|$. In total, we used $\Order(|D|)$ data structure operations, leading to a running time of $\Oo(|T| \cdot |(S+T)\setminus S| \cdot \log u)$.

        \item \Op{cap}$(d)$: 
        We delete all vectors $s \in S$ with at least one entry that is larger than~$d$ as follows. We enumerate all $(m-1)$-tuples $(s_1, \dots, s_{m-1}) \in [u]^{m-1}$. Recall that the vectors are stored in lexicographical order. Thus, the set of vectors $(s_0, s_1, \dots, s_{m-1})$ where~$s_0$ ranges over~$[u]$ and $s_1, \dots, s_{m-1}$ are fixed, is represented by a length-$u$ substring of $x_{\phi(S)}$. Specifically, the substring $x_{\phi(S)}\intervalcc{\phi(0, s_1, \dots, s_{m-1})}{\phi(u-1, s_1, \dots, s_{m-1})}$. We~distinguish two cases: If $\max_{i=1}^{m-1} s_i > d$, then the entire substring is replaced with~$0^u$. Otherwise, we retain its length-$(d+1)$ prefix and replace its suffix is by $0^{u-d-1}$.
        It takes~$\Oo(u^{m-1})$ \Op{concat} and \Op{split} operations to perform these modifications, running in total time $\Order(u^{m-1} \cdot \log u)$.
    \end{itemize}
    \smallskip
    Again, following \cref{lem:dynamic_string_ds} the deterministic running time of the operations differs by replacing $\log u$ with $\log^{1+\order(1)} u$.
\end{proof}

Based on \cref{lem:generalized-bounded-insertions} and \ref{lem:generalized-set-ds}, we show the following generalization of our main result.

\ManyMachinesAlgorithm*
\begin{proof}
    Analogous to \cref{theo:mtpt}, we use the algorithm:
    We initialize the data structure from \cref{lem:generalized-set-ds} (used with $u = P + 1$) with \Op{init}$(S_0)$.
    For~all~$j \gets 0, \dots, n-1$, we repeatedly use the operations \Op{sum}$(\set{0, p_j \cdot e_0, \dots, p_j \cdot e_{m-1}})$ and \Op{cap}$(d_j)$ to compute the sets~$S_{j+1}'$ and $S_{j+1}$. 
    We return the optimal value contained in $S_n$ as described in \cref{theo:mtpt}.
    Our algorithm is correct due to \cref{obs:generalized-s-sets}.
    Finally, using \cref{lem:generalized-set-ds,lem:generalized-bounded-insertions} and the assumption $n \leq P$, we can bound the dominant term of the running time by
    \begin{equation*}
        \Order\parens*{\sum_{j \in [n]} (m \cdot |S_{j+1}' \setminus S_j| \cdot \log P + P^{m-1} \log P)} = \Order(P^m \log P).
    \end{equation*}
    (Recall that $m$ is constant.) Thus, we obtain a randomized running time of $\Order(P^m \log P)$, and similarly a deterministic running time of~$\Order(P^m \log^{1+\order(1)} P)$.
\end{proof}



\printbibliography

\end{document}